\documentclass[12pt]{amsart}
\usepackage{amsmath}
\usepackage{amsfonts}
\usepackage{amsthm}
\usepackage[mathscr]{eucal}
\usepackage{amssymb}
\input{Qcircuit}

\newtheorem{thm}{Theorem}[section]

\newtheorem{cor}[thm]{Corollary}

\newtheorem{lemma}[thm]{Lemma}
\newtheorem{prop}[thm]{Proposition}

\theoremstyle{definition}
\newtheorem{defn}[thm]{Definition}
\newtheorem{remark}[thm]{Remark}

\newcommand{\bb}[1]{\mathbb{#1}}
\newcommand{\cl}[1]{\mathcal{#1}}

\renewcommand{\ket}[1]{|#1\rangle}
\renewcommand{\bra}[1]{\langle#1|}
\newcommand{\HH}{\mathcal{H}}
\newcommand{\RR}{\mathcal{R}}
\newcommand{\NN}{\mathbb{N}_0}
\newcommand{\ZZ}{\mathbb{Z}}

\newcommand{\sqq}{{\textstyle{\frac{1}{\sqrt 2}}}}

\newcommand{\Tr}{\mbox{Tr}}
\newcommand{\nc}{\operatorname{nc}}

\begin{document}

\title[Perfect Embezzlement of Entanglement]{Perfect Embezzlement of Entanglement}
\author[R.~Cleve]{Richard Cleve}
\address[R.~Cleve]{\newline Institute for Quantum Computing and School of Computer Science, University of Waterloo; Canadian Institute for Advanced Research}
\email{cleve@uwaterloo.ca}
\author[L.~Liu]{Li Liu}
\address[L.~Liu]{\newline Institute for Quantum Computing and School of Computer Science, University of Waterloo}
\email{l47liu@uwaterloo.ca}
\author[V.~I.~Paulsen]{Vern I.~Paulsen}
\address[V.~I.~Paulsen]{\newline Institute for Quantum Computing and Department of Pure Mathematics, University of Waterloo}
\email{vpaulsen@uwaterloo.ca}

\thanks{}

\begin{abstract}
Van Dam and Hayden introduced a concept commonly referred to as \emph{embezzlement}, where, for any entangled quantum state $\phi$, there is an entangled catalyst state $\psi$, from which a high fidelity approximation of $\phi \otimes \psi$ can be produced using only local operations.
We investigate a version of this where the embezzlement is \emph{perfect} (i.e., the fidelity is~1).
We prove that perfect embezzlement is impossible in a tensor product framework, even with infinite-dimensional Hilbert spaces and infinite entanglement entropy.
Then we prove that perfect embezzlement \emph{is} possible in a commuting operator framework.
We prove this using the theory of C*-algebras and we also provide an explicit construction.
Next, we apply our results to analyze perfect versions of a nonlocal game introduced by Regev and Vidick.
Finally, we analyze the structure of perfect embezzlement protocols in the commuting operator model, showing that they require infinite-dimensional Hilbert spaces.
\end{abstract}

\maketitle


\section{Introduction}

It is well known that an entangled quantum state cannot be produced by local operations alone.
Van Dam and Hayden~\cite{vanDamH03} proposed a method that, in a certain sense, appears to produce \textit{additional} entanglement by local operations.
They showed that, for any entangled state $\phi$ and $\epsilon > 0$, starting with a special entangled \textit{catalyst} state $\psi$, applying local operations, can produce a state that approximates $\phi \otimes \psi$ within fidelity $1-\epsilon$.
Although the entanglement entropy of the state produced cannot exceed that of $\psi$, when $\epsilon$ is small, it is difficult to distinguish between the state produced and $\phi \otimes \psi$.
The name \textit{embezzlement} reflects the fact that the protocol ``steals" entanglement from $\psi$ in order to produce entanglement elsewhere, but in a manner that is difficult to detect.

In the method of~\cite{vanDamH03}, fidelity $1-\epsilon$ can be attained for any $\epsilon > 0$, using a catalyst $\psi$ with entanglement entropy $O(\log(1/\epsilon))$.
Moreover, it is shown in~\cite{vanDamH03} that the entanglement entropy of the catalyst must be $\Omega(\log(1/\epsilon))$ to attain this fidelity.
Thus, high fidelity embezzlement requires a large amount of entanglement to begin with.

We consider the question: what kinds of embezzlement are possible when the amount of entanglement in $\psi$ is allowed to be \textit{infinite}?
The aforementioned results do not rule out perfect (i.e., fidelity 1) embezzlement in such cases.
On the other hand, the catalytic states $\psi_{\epsilon}$ in~\cite{vanDamH03} do not converge to a valid quantum vector state as $\epsilon$ approaches 0.
This question provides a setting in which the consequences of notions of infinite entanglement can be explored.

We first show that in the \textit{tensor product framework}, where catalytic states are in the tensor product of two Hilbert spaces, perfect embezzlement is impossible, even if the spaces are infinite dimensional and the entanglement entropy is infinite.

Next, we consider a \textit{commuting operator framework}, where the notion of ``local" is formalized differently: there is one joint Hilbert space, accessible to both Alice and Bob; however, the operations that Alice performs on this space must commute with those of Bob.
This formalism is used in quantum field theory 
(see~\cite{Tsirelson1993,Navascues2008,scholz2008,DohertyLTW2008,junge2011,Fritz2012} for more discussion about this framework and its relationship with the tensor product framework).

A natural adaptation of the commuting operator framework to the setting of embezzlement is the following.
The catalytic state $\psi$ is in a jointly accessibe Hilbert space, that we refer to as the \textit{resource space}~$\mathcal R$.
There are also two additional Hilbert spaces: $\cl H_A$, accessible to Alice only; and $\cl H_B$, accessible to Bob only.
The goal of the protocol is to transform a product state to an entangled state in $\cl H_A \otimes \cl H_B$ while using $\psi$ catalytically, and using operators that are commuting in the following sense.
Alice can apply a unitary operator on $\HH_A \otimes \RR$ and Bob can apply a unitary operator on $\HH_B \otimes \RR$; however, $U_A \otimes I_{\HH_B}$ and $I_{\HH_A} \otimes U_B$ must commute on $\HH_A \otimes \RR \otimes \HH_B$, as illustrated in Figure~\ref{fig:commuting-structure}.
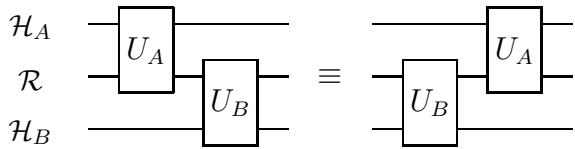
\begin{figure}[!ht]
\begin{center}
\setlength{\unitlength}{0.14mm}
\begin{picture}(55,140)(0,0)
\linethickness{0.5pt}
\put(0,10){\makebox(20,14){\small $\HH_B$}}
\put(0,60){\makebox(20,14){\small $\mathcal R$}}
\put(0,110){\makebox(20,14){\small $\HH_A$}}
\end{picture}
\begin{picture}(190,140)(0,0)
\linethickness{0.5pt}
\put(0,120){\line(1,0){30}}
\put(80,120){\line(1,0){110}}
\put(0,70){\line(1,0){30}}
\put(0,20){\line(1,0){110}}
\put(160,20){\line(1,0){30}}
\put(110,5){\framebox(50,80){$U_B$}}
\put(80,70){\line(1,0){30}}
\put(30,55){\framebox(50,80){$U_A$}}
\put(160,70){\line(1,0){30}}
\end{picture}
\begin{picture}(60,140)(0,0)
\put(0,60){\makebox(60,20){$\equiv$}}
\end{picture}
\begin{picture}(190,140)(0,0)
\linethickness{0.5pt}
\put(0,20){\line(1,0){30}}
\put(80,20){\line(1,0){110}}
\put(0,70){\line(1,0){30}}
\put(0,120){\line(1,0){110}}
\put(160,120){\line(1,0){30}}
\put(30,5){\framebox(50,80){$U_B$}}
\put(80,70){\line(1,0){30}}
\put(110,55){\framebox(50,80){$U_A$}}
\put(160,70){\line(1,0){30}}
\end{picture}
\end{center}
\caption{\small Commuting operator framework 
as a circuit diagram.}
\label{fig:commuting-structure}
\end{figure}

We focus on the problem of embezzling a Bell state of the form $\sqq\ket{0}\otimes\ket{0}+\sqq\ket{1}\otimes\ket{1}$ (though our methodology adapts to more general states).
In this case, $\cl H_A = \cl H_B = \bb C^{2}$.
A {\it perfect embezzlement protocol} consists of 
a resource space $\RR$, 
a catalytic state $\psi \in \cl \RR$, and
commuting unitary operators $U_A$ and $U_B$, such that 
\begin{align}\label{eq:perfect-embezzlement}
(U_A \otimes I_{\HH_B})&(I_{\HH_A} \otimes U_B) \ket{0} \otimes \psi \otimes \ket{0} \\
&= \textstyle{\frac{1}{\sqrt{2}}}\ket{0} \otimes \psi \otimes \ket{0} + 
\textstyle{\frac{1}{\sqrt{2}}}\ket{1} \otimes \psi \otimes \ket{1}. \nonumber
\end{align}

We show that, in this commuting operator framework, a perfect embezzlement protocol exists, where the resource space is a countably infinite dimensional (i.e., separable) Hilbert space.
We show this in two ways: one is a simple existence proof, based on the theory of C*-algebras, which does not yield explicit unitary operations; the other is by an explicit construction.

Next, we consider \textit{coherent embezzlement}, which was introduced in~\cite{RV2013} (where it is referred to as $T_2$) and is a refinement of coherent state exchange, introduced in~\cite{LeungTW13}.
Coherent embezzlement is related to embezzlement but has the property that it is operationally testable in a sense similar to that of nonlocal games (whereas embezzlement itself does not have this property).
We give reductions between perfect embezzlement and perfect coherent embezzlement to prove that perfect coherent embezzlement is impossible in the tensor product framework; whereas it is possible in the commuting operator framework.

Finally, we prove a theorem concerning the structure of pairs of unitaries that achieve perfect embezzlement in terms of properties of their {\it constituent operators}. We show that at least one of these operators must contain a {\it non-unitary isometry}, a term that we will define later. Since non-unitary isometries do not exist in finite dimensions, this implies that perfect embezzlement in the commuting-operator model cannot be achieved with a finite dimensional resource space.


\section{Perfect embezzlement is impossible in a tensor product framework}

In~\cite{vanDamH03}, it is proved that, for any protocol that embezzles within fidelity $1 - \epsilon$, the entanglement entropy of the catalyst must be $\Omega(\log(1/\epsilon))$.
It follows that perfect embezzlement is impossible with finite-dimensional entanglement in the tensor product framework.
Here, we extend this impossibility result to tensor products of arbitrary Hilbert spaces (where the dimension of the spaces and entanglement entropy can be infinite).

In the tensor product framework, the resource space is of the form $\mathcal{R} = \mathcal{R}_A \otimes \mathcal{R}_B$, where $\mathcal{R}_A$ and $\mathcal{R}_B$ are arbitrary Hilbert spaces.
Alice has access to $\mathcal{H}_A \otimes \mathcal{R}_A$ and Bob has access to $\mathcal{H}_B \otimes \mathcal{R}_B$.
Alice and Bob can each apply any unitary operation to the registers that they have access to, as illustrated in Figure~\ref{fig:tensor} (left), where the input state is $\ket{0}\otimes \psi \otimes \ket{0}$, for some state $\psi \in \mathcal{R}_A \otimes \mathcal{R}_B$.
The protocol performs perfect embezzlement if its output state is
$\textstyle{\frac{1}{\sqrt{2}}}\ket{0} \otimes \psi \otimes \ket{0} + 
\textstyle{\frac{1}{\sqrt{2}}}\ket{1} \otimes \psi \otimes \ket{1}$.

We also define a potentially stronger model, that we refer to as \textit{embezzlement with ancillas}, which includes the possibility of Alice and Bob employing additional registers as part of their protocol, as illustrated in Figure~\ref{fig:tensor} (right).
\begin{figure}[!ht]
\centerline{
\Qcircuit @C=1em @R=.7em{
\\
\lstick{\mathcal{H}_A \ \ \ } & \multigate{1}{U_A}      & \qw \\
\lstick{\mathcal{R}_A \ \ \ } & \ghost{U_{A}}      & \qw \\
\lstick{\mathcal{R}_B \ \ \ } & \multigate{1}{U_B} & \qw \\
\lstick{\mathcal{H}_B \ \ \ } & \ghost{U_{B}}      & \qw 
}
\hspace*{40mm}
\Qcircuit @C=1em @R=.7em{
\lstick{\mathcal{G}_A \ \ \ } & \multigate{2}{U_A} & \qw \\
\lstick{\mathcal{H}_A \ \ \ } & \ghost{U_{A}}      & \qw \\
\lstick{\mathcal{R}_A \ \ \ } & \ghost{U_{A}}      & \qw \\
\lstick{\mathcal{R}_B \ \ \ } & \multigate{2}{U_B} & \qw \\
\lstick{\mathcal{H}_B \ \ \ } & \ghost{U_{B}}      & \qw \\
\lstick{\mathcal{G}_B \ \ \ } & \ghost{U_{B}}      & \qw 
}
}
\caption{\small Circuit diagram for embezzlement (left) and embezzlement with ancillas (right) in the tensor product framework.
Registers $\mathcal{R}_A$ and $\mathcal{R}_B$ contain a bipartite resource state that must be used catalytically.
Registers $\cl H_A = \cl H_B = \bb C^{2}$ are intialized to state $\ket{00}$ and are transformed to state $\sqq\ket{00}+\sqq\ket{11}$.
Registers $\mathcal{G}_A$ and $\mathcal{G}_B$ are ancillas, whose initial state is unentangled, but they need not be used catalytically.
}\label{fig:tensor}
\end{figure}
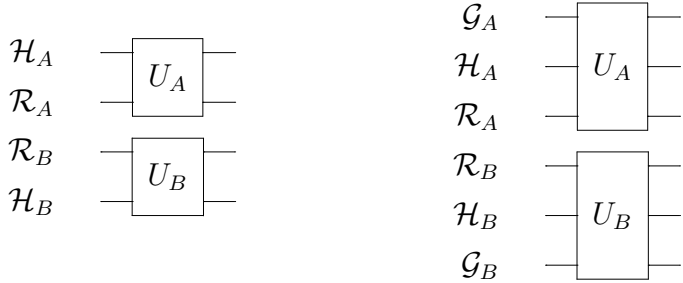
The input to the circuit is of the form $\gamma_A \otimes \ket{0} \otimes \psi \otimes \ket{0} \otimes \gamma_B$, where $\psi \in \mathcal{R}_A \otimes \mathcal{R}_B$ is the catalyst state, and $\gamma_A \in \mathcal{G}_A$ and $\gamma_B \in \mathcal{G}_B$ are the initial states of Alice and Bob's respective ancilla registers, $\mathcal{G}_A$ and $\mathcal{G}_B$ (which can be infinite dimensional).
If we express the Hilbert space as 
$(\mathcal{H}_A \otimes \mathcal{H}_B) \otimes (\mathcal{R}_A \otimes \mathcal{R}_B) \otimes (\mathcal{G}_A \otimes \mathcal{G}_B)$
then the input state can be written as 
$\ket{00}\otimes\psi\otimes (\gamma_A \otimes \gamma_B)$.
The protocol \textit{performs perfect embezzlement} if and only if the output state is of the form 
\begin{align}
\bigl(\sqq\ket{00}+\sqq\ket{11}\bigr)\otimes\psi\otimes\gamma_{AB},
\end{align}
for some state $\gamma_{AB} \in \mathcal{G}_A \otimes \mathcal{G}_B$.

\begin{thm}\label{thm:no-spatial-embezzlement}
Perfect embezzlement is impossible in the tensor product framework, even if Alice and Bob are allowed to use ancillas.
\end{thm}
\begin{proof}
The proof is a straightforward application of the Schmidt decomposition for vectors in tensor products of arbitrary Hilbert spaces.
For arbitrary (not necessarily separable) 
 Hilbert spaces $\mathcal{H}$ and $\mathcal{K}$ and any $\phi \in \mathcal{H} \otimes \mathcal{K}$, it is possible to express
\begin{align}
\phi = \sum_{j = 0}^{\infty} \alpha_j\, u_j\otimes v_j, 
\end{align}
where $\alpha_j \ge 0$, $\sum_{j = 0}^{\infty} |\alpha_j|^2 = 1$, $ \alpha_j \ge \alpha_{j+1}$, $u_0, u_1, \dots $ are orthonormal vectors in $\mathcal{H}$, and $v_0, v_1, \dots$ are orthonormal vectors in $\mathcal{K}$. Moreover, given these conditions, the coefficients $\alpha_j$ are unique.
For the convenience of the reader, we include a proof of this in Appendix~\ref{appendix:Schmidt}.

Now taking a Schmidt decomposition of
$\gamma_A \otimes \ket{0} \otimes \psi \otimes \ket{0} \otimes \gamma_B$, with respect to 
$\big(\mathcal{G}_A \otimes \mathcal{H}_A \otimes \mathcal{R}_A \big) \otimes \big(\mathcal{G}_B \otimes \mathcal{H}_B \otimes \mathcal{R}_B \big)$, we obtain Schmidt coefficients $\alpha_j$.

Suppose that a perfect embezzlement protocol exists.
Then, since $U_A$ and $U_B$ are local unitaries, the Schmidt coefficients of the initial state $\ket{00} \otimes \psi \otimes \gamma_A \otimes \gamma_B$ must be the same as those of the final state $(\sqq\ket{00} + \sqq\ket{11})\otimes \psi \otimes \gamma_{AB}$.
But this is a contradiction, since the largest Schmidt coefficient of the input state is $\alpha_0$ (which is nonzero) and the largest Schmidt coefficient of the output state is at most $\sqq\alpha_0$.
Therefore, there is no perfect embezzlement protocol in the tensor product framework.
\end{proof}


\section{Perfect embezzlement is possible in a commuting operator framework}\label{approxtoperfect}

In this section we show that, since one can approximately embezzle a Bell state to any level of precision in finite dimensions (by the results of~\cite{vanDamH03}), one can perfectly embezzle in infinite
dimensions in the commuting operator framework.  
Readers unfamiliar with the theory of C*-algebras might prefer to read our primer on C*-algebras in Appendix~\ref{appendix:cstarprimer} before tackling this section.
At the end of the section, we explain how to generalize the technique to more general entangled states.

We begin by showing that each commuting operator framework, where $\cl H_A = \cl H_B = \bb C^2$, yields a set of eight operators on the resource space.  To study the most general commuting framework, it is natural to consider the relations that any such set of operators must satisfy and look for a ``universal" model for such sets of operators.

We will show that the eight operators arising from a commuting operator framework are always a representation of a certain C*-algebra and that the catalyst vector yields a state on this C*-algebra. We will show that the commuting operator framework together with the catalyst vector achieves perfect embezzlement of a Bell state if and only if the state on this C*-algebra induced by the catalyst vector satisfies a set of four equations.

In this manner the question of whether or not one can perfectly embezzle a Bell state is reduced to a question about the existence of a state on this C*-algebra that satisfies our four equations.

Finally, we show that perfect embezzlement of a Bell state is possible in the commuting operator framework by showing the existence of such a state. 

The ``universal"  C*-algebra that one needs was first introduced by L.G. Brown~\cite{Br}, who referred to it as the {\it universal C*-algebra of a non-commutative unitary} for reasons that will, hopefully, be clear.  Our viewpoint shows that in a certain sense questions about embezzlement can be interpreted as questions about states on these particular quantum group C*-algebras. We think that this perspective is new and should lead to interesting links between these two areas.

Let's return to the scenario of Figure~\ref{fig:commuting-structure}. Alice's unitary operation, $U_A: \bb C^2 \otimes \cl R \to \bb C^2 \otimes \cl R$ can be represented by a $2 \times 2$ matrix of operators on $\cl R$, 
\[U_A= \begin{pmatrix} U_{00} & U_{01} \\ U_{10} & U_{11} \end{pmatrix}= \sum_{i,j=0}^1 \ket{i} \bra{j} \otimes U_{ij} \]
 where $U_A( \ket{j} \otimes h) = \sum_{i=0}^1 \ket{i} \otimes U_{ij} h$. In this case, 
\[ U_A^* = \begin{pmatrix} U_{00}^* & U_{10}^* \\ U_{01}^* & U_{11}^* \end{pmatrix} \]
and the fact that $U_A$ is unitary can be expressed by eight equations involving these operators that are best expressed as
\begin{align}\label{eq:8-eqns} 
U_A^*U_A = \begin{pmatrix} I_{\cl R} & 0\\ 0 & I_{\cl R} \end{pmatrix} = U_AU_A^*,
\end{align}
where we apply the usual rules of matrix multiplication, being careful to remember that since the entries of $U_A$ are operators, not numbers, they need not commute. We also recall that when $\cl R$ is infinite dimensional, then it is necessary that both $U_A^*U_A$ and $U_AU_A^*$ be the identity to guarantee that $U_A$ is unitary.

Finally, in the special case that $\dim(\cl R) =1$ so that these entries are numbers, then we are back to the usual case of a $2 \times 2$ complex unitary matrix.

Conversely, if we let $U_A$ be any $2 \times 2$ matrix of operators on $\cl R$ that satisfies Eq.~\eqref{eq:8-eqns} then $U_A$ will define a unitary on $\bb C^2 \otimes \cl R$.

Similarly, Bob's unitary $U_B: \cl R \otimes \bb C^2 \to \cl R \otimes \bb C^2$ is represented by a $2 \times 2$ matrix of operators on $\cl R$, $U_B= ( V_{ij})$, whose entries satisfy the same eight equations.

Finally, to have a commuting operator framework as in Figure~\ref{fig:commuting-structure}, we need $(U_A \otimes I_2)(I_2 \otimes U_B) = (I_2 \otimes U_B)(U_A \otimes I_2)$. The following proposition translates this condition into equations involving the operator entries.

\begin{prop} Let $U_{ij}, V_{kl}, 0 \le i,j,k,l \le 1$ be operators on the Hilbert space $\cl R$ such that $U_A=( U_{ij}),$ and $U_B= (V_{kl})$ are unitaries.  Then $(U_A \otimes I_2)(I_2 \otimes U_B)
= (I_2 \otimes U_B)(I_2 \otimes U_A)$ if and only if $U_{ij}V_{kl} = V_{kl}U_{ij}$ and $U_{ij}^*V_{kl} = V_{kl}U_{ij}^*$ for all $i,j,k,l$.
\end{prop}
\begin{proof} We have that 
\[(U_A \otimes I_2)(I_2 \otimes U_B)( \ket{j} \otimes h \otimes \ket{l} ) = \sum_{i,k=0}^1 \ket{i} \otimes U_{ij}V_{kl}h \otimes \ket{k},\]
and similarly,
\[(I_2 \otimes U_B)(U_A \otimes I_2) ( \ket{j} \otimes h \otimes \ket{l}) = \sum_{i,k=0}^1 \ket{i} \otimes V_{kl}U_{ij}h \otimes \ket{k}.\] 
Thus, we see that $(U_A \otimes I_2)(I_2 \otimes U_B) = (I_2 \otimes U_B)(U_A \otimes I_2)$ is equivalent to $U_{ij}V_{kl} = V_{kl}U_{ij}$, for all $i,j,k,l.$

However, if an invertible operator commutes with another operator, then its inverse also commutes with that operator. Hence, $(U_A \otimes I_2)^{-1} = (U_A^* \otimes I_2)$ commutes with $(I_2 \otimes U_B)$ and this is equivalent to  $U_{ij}^*V_{kl}= V_{kl}U_{ij}^*$, for all $i,j,k,l.$
\end{proof}

The above equations are generally summarized by saying that the set of operators $\{ U_{ij} \}$ {\it *-commutes} with the set $\{ V_{kl} \}$.
Thus, having a commuting operator framework is equivalent to having two unitaries $U_A=(U_{ij})$, and $U_B=(V_{kl})$ whose entries *-commute.

We wish to study ``universal" properties of $2 \times 2$ matrices of operators $(U_{ij})$ that give rise to a unitary. To do this we begin with a unital *-algebra $\cl U_2$ with generators, denoted $1$ and  $u_{ij}, \, 0 \le i,j \le 1$, subject to the eight equations, 
\begin{align} 
\begin{pmatrix} u_{00} & u_{01} \\ u_{10} & u_{11} \end{pmatrix} 
\begin{pmatrix} u_{00}^* & u_{10}^* \\ u_{01}^* & u_{11}^* \end{pmatrix} 
=\begin{pmatrix} u_{00}^* & u_{10}^* \\ u_{01}^* & u_{11}^* \end{pmatrix} 
\begin{pmatrix} u_{00} & u_{01} \\ u_{10} & u_{11} \end{pmatrix} 
= \begin{pmatrix} 1 & 0 \\ 0 & 1 \end{pmatrix}.
\end{align}
 Thus, whenever there is a Hilbert space $\cl H$ and four operators, $U_{ij}$ on that space such that the operator-matrix $U=(U_{ij})$ defines a unitary operator on $\cl H \otimes \bb C^2$, then there is a *-homomorphism,
\[ \pi: \cl U_2 \to B(\cl H) \text{ with } \pi(u_{ij}) = U_{ij}.\]
For $x \in \cl U_2$ he sets $\|x\|= \sup \{ \|\pi(x)\| \colon \pi \text{ a *-homomorphism} \}$, where the supremum is taken over all Hilbert spaces and all $\pi$'s as above. This defines a norm on $\cl U_2$ and that the completion is a C*-algebra, we shall denote $U_{\nc}(2)$. The subscript {$\nc$} stands for ``non-commuting" and is intended to remind us that the generators $u_{ij}$ do not commute. (This approach generalizes naturally to $d \times d$ matrices of operators,  for $d > 2$, where the C*-algebra is denoted as $U_{\nc}(d)$.)

Note that in the commuting operator framework, the set $U_{ij}$ and the set $V_{kl}$ each gives rise to a *-homomorphism and that these two *-homomorphisms commute. Thus, it is not hard to see that we have a one-to-one correspondence between commuting operator frameworks and *-homomorphisms of $U_{\nc}(2) \otimes U_{\nc}(2)$ into $B(\cl R)$.  Since we want to consider all commuting operator frameworks, we are lead to study $U_{\nc}(2) \otimes_{\max} U_{\nc}(2)$.  

The study of states on this algebra turns out to be closely related to embezzlement constructions as the following result shows.

\begin{thm} \label{statecharembezzlement} There exists a perfect embezzlement protocol in the commuting operator framework if and only if there exists a state 
$s: U_{\nc}(2) \otimes_{\max} U_{\nc}(2) \to \bb C$ such that
\begin{itemize}
\item $s(u_{00} \otimes u_{00}) = \frac{1}{\sqrt{2}}$, \vspace*{1mm}
\item $s(u_{10} \otimes u_{00}) = \,0$, \vspace*{1mm}
\item $s(u_{00} \otimes u_{10}) = \,0$, \vspace*{1mm}
\item $s(u_{10} \otimes u_{10}) = \frac{1}{\sqrt{2}}$.
\end{itemize}
\end{thm}
\begin{proof}
First assume that a perfect embezzlement protocol exists in a commuting operator framework.
Let $\cl R$ be a Hilbert space, let $\psi \in \cl R$ be a unit vector, let $U_A= \big( U_{ij} \big)$ and $U_B= \big( V_{kl} \big)$ be unitaries on $\bb C^2 \otimes \cl H$ and $\cl H \otimes \bb C^2$, respectively, such that $U_A \otimes I_2$ commutes with $I_2 \otimes U_B$ and let
$\psi$ be a catalyst vector for perfect embezzlement of a Bell state,
Define $\pi: U_{\nc}(2) \otimes_{\max} U_{\nc}(2) \to B(\cl H)$ to be the *-homomorphism defined by $\pi(u_{ij} \otimes 1) = U_{i,j}, \, \pi(1 \otimes u_{kl}) = V_{kl}$. Since $\psi$ is a catalyst vector,
\[ (U_A \otimes I_2)(I_2 \otimes U_B) ( \ket{0} \otimes \psi \otimes \ket{0}) = \textstyle{\frac{1}{\sqrt{2}}}( \ket{0} \otimes \psi \otimes \ket{0} + \ket{1} \otimes \psi \otimes \ket{1}). \]

Now define a state on $U_{\nc}(2) \otimes_{\max} U_{\nc}(2)$  by $s(x) = \langle \pi(x) \psi, \psi \rangle$. 
We have that
\begin{align*}
\big( I_2 \otimes U_B\big) \big( U_A \otimes I_2 \big) \ket{0} \otimes \psi \otimes \ket{0} & = \sum_{i,j=0}^1 \ket{i} \otimes U_{i,0}V_{j,0} \psi \otimes \ket{j} 
\\ & = \textstyle{\frac{1}{\sqrt{2}}}\ket{0} \otimes \psi \otimes \ket{0} + \textstyle{\frac{1}{\sqrt{2}}}\ket{1} \otimes \psi \otimes \ket{1},
\end{align*}
which is equivalent to 
\[ U_{00}V_{00} \psi = U_{10}V_{10} \psi = \textstyle{\frac{1}{\sqrt{2}}}\psi \text{\ \  and \ \ } U_{00} V_{10} \psi = U_{10}V_{00} \psi = 0 .\]
From these equations, it follows that the state $s$ satisfies the four conditions.

Conversely, assume that $s: U_{\nc}(2) \otimes_{\max} U_{\nc}(2) \to \bb C$ is a state that satisfies the 4 conditions.
 Let $\pi_s: U_{\nc}(2) \otimes_{\max} U_{\nc}(2) \to B(\cl H_s)$  and $\psi \in \cl H_s$ be the GNS representation of the state so that $s(x) = \langle \pi_s(x) \psi, \psi \rangle$. If we define $U_A: \bb C^2 \otimes \cl H \to \bb C^2 \otimes \cl H$ by $U_A= \big( \pi(u_{ij} \otimes 1) \big)$ and $U_B: \cl H \otimes \bb C^2 \to \cl H \otimes \bb C^2$ by $U_B = \big( \pi(1 \otimes u_{i,j}) \big)$, then $U_A \otimes I_2$ commutes with $I_2 \otimes U_B$.
The operator on the direct sum of four copies of $\cl H$ given by
\[ \begin{pmatrix} U_{00}V_{00} & U_{01}V_{00} & U_{00}V_{01} & U_{01}V_{01} \\ U_{10}V_{00} & U_{11}V_{00} & U_{10}V_{01} & U_{11}V_{01}\\ U_{00} V_{10} & U_{01}V_{10} & U_{00}V_{11} & U_{01}V_{11}\\ U_{10}V_{10} & U_{11}V_{10} & U_{10} V_{11} & U_{11}V_{11} \end{pmatrix} \]
is unitary.

Hence,
\begin{align*} 1 &= |\langle U_{00}V_{00} \psi, \psi \rangle|^2 + |\langle U_{10}V_{10} \psi, \psi \rangle |^2 \le \| U_{00} V_{00} \psi \|^2 + \|U_{10}V_{10}\psi \|^2  \\ &\le
\| U_{00} V_{00} \psi \|^2 + \|U_{10}V_{10}\psi \|^2 + \|U_{00}V_{10}\psi \|^2 + \|U_{10} V_{00} \psi \|^2 =1, \end{align*}
from which it follows that
\[ U_{00}V_{00} \psi = U_{10}V_{10} \psi = \textstyle{\frac{1}{\sqrt{2}}}\psi  \text{\ \  and \ \,} U_{00} V_{10} \psi = U_{10}V_{00} \psi = 0 .\]

Thus,
\begin{align*}
\big( I_2 \otimes U_B\big) \big( U_A \otimes I_2 \big) \ket{0} \otimes \psi \otimes \ket{0} & = \sum_{i,j=0}^1 \ket{i} \otimes U_{i,0}V_{j,0} \psi \otimes \ket{j} 
\\ & = \textstyle{\frac{1}{\sqrt{2}}}\ket{0} \otimes \psi \otimes \ket{0} + \textstyle{\frac{1}{\sqrt{2}}}\ket{1} \otimes \psi \otimes \ket{1}
\end{align*}
and we have a perfect embezzlement protocol. \end{proof}

Thus, we have proven that perfect embezzlement in the commuting operator framework is equivalent to the existence of a state on $U_{\nc}(2) \otimes_{\max} U_{\nc}(2)$ that satisfies the four equations above. We now prove that such a state exists.

\begin{thm}\label{thm:exists-4-state}  There exists a state $s: U_{\nc}(2) \otimes_{\max} U_{\nc}(2) \to \bb C$ that satisfies the four equations of the previous theorem and consequently perfect embezzlement is possible in the commuting operator framework.
\end{thm}
\begin{proof} By the results of \cite{vanDamH03},  we have finite dimensional Hilbert spaces $H_n$ unit vectors $h_n
\in H_n$ and unitary operators $U_n,  V_n$ on $H_n \otimes \bb C^2$,  such that
$(U_n \otimes I_2) (I_2 \otimes V_n)(\ket{0} \otimes h_n \otimes \ket{0}) -
\frac{1}{\sqrt{2}}(\ket{0} \otimes h_n \otimes \ket{0} + \ket{1} \otimes h_n \otimes
\ket{1})$ has norm less than $1/n$.

These operators induce *-homomorphisms, $\pi_n: U_{\nc}(2) \otimes_{\max} U_{\nc}(2) \to B(H_n)$ and states $s_n: U_{\nc}(2) \otimes_{\max} U_{\nc}(2) \to \bb C$ defined by $s_n(x) = \langle\pi_n(x) h_n , h_n \rangle$. These states satisfy:
\begin{itemize}
\item $\bigl| s_n( u_{00} \otimes u_{00}) - \frac{1}{\sqrt{2}} \bigr| < \frac{1}{n}$, \vspace*{1mm}
\item $| s_n(u_{10} \otimes u_{00}) | < \frac{1}{n}$, \vspace*{1mm}
\item $| s_n( u_{00} \otimes u_{10}) | < \frac{1}{n}$, \vspace*{1mm}
\item $ \bigl| s_n(u_{10} \otimes u_{10}) - \frac{1}{\sqrt{2}} \bigr| < \frac{1}{n}$.
\end{itemize}
Now by the fact that the state space of any unital C*-algebra is compact in the weak*-topology, we may take a limit point $s$ of this sequence of states. Since the value of $s(u_{i,j} \otimes u_{k,l})$ must be a limit of the values of $s_n$ on these same elements, $s$ will be a state on $U_{\nc}(2) \otimes_{\max} U_{\nc}(2)$ that satisfies the 4 conditions exactly. 
\end{proof}

\begin{remark}  From \cite{vanDamH03}, each of the states, $s_n$ appearing in the above proof is actually a state on $U_{\nc}(2) \otimes_{\min} U_{\nc}(2)$. Hence, by taking a limit point, we obtain a state $s: U_{\nc}(2) \otimes_{\min} U_{\nc}(2) \to \bb C$ that satisfies the 4 equations of Theorem~\ref{statecharembezzlement}. If we apply the GNS construction or any other method to represent it as $s(x) = \langle \pi(x) \psi, \psi \rangle$ on some Hilbert space $\cl H$, where $\pi: U_{\nc}(2) \otimes_{\min} U_{\nc}(2) \to B(\cl H)$, then the representation $\pi$ and the Hilbert space cannot decompose as a tensor product. Otherwise we would achieve perfect embezzlement in a tensor product framework.  Hence, we obtain an example of a state on a minimal tensor product, such that it cannot be represented using a *-homomorphism that is a spatial tensor product.  In fact, no state on $U_{\nc}(2) \otimes_{\min} U_{\nc}(2)$ that satisfies just those 4 equations can have a spatial tensor product representation.
\end{remark}

\begin{remark}  The coefficients that appear in Theorem~\ref{statecharembezzlement} are a consequence of the fact that we are embezzling a Bell state. If we wish instead for a perfect embezzlement protocol of a more general vector state, say, of the form
$\sum_{i,j=0}^{d-1} \alpha_{ij} \ket{i} \otimes \ket{j}$,
then this is equivalent to the existence of a state $s$ on $U_{\nc}(d) \otimes_{\max} U_{\nc}(d)$ satisfying
$s(u_{i0} \otimes u_{j0}) = \alpha_{ij}$, for $0 \le i,j < d$. Moreover, it is shown in \cite{vanDamH03} that every vector in $\bb C^d \otimes \bb C^d$ can be approximately embezzled in a finite dimensional scenario.
Therefore, arguing as above, there is always a state $s$ on $U_{\nc}(d) \otimes_{\min} U_{\nc}(d)$ satisfying the $d^2$ equations.
\end{remark}


\section{Explicit construction of a perfect embezzlement protocol in a commuting-operator framework}\label{sec:explicit}

The previous section proves the existence of a perfect embezzlement protocol, but without constructing one explicitly.
Some of the steps of the proof are nonconstructive. 
In Theorem~\ref{thm:exists-4-state}, an abstract state is obtained by invoking an existence theorem using the weak*-compactness of the set of all states; moreover, in Theorem~\ref{statecharembezzlement}, the Hilbert space is obtained by applying the GNS representation of the state, which is based on the completion of an abstract C*-algebra.
In this section, we give an explicit commuting-operator protocol for perfect embezzlement.
We explain the technique for Bell states, and, at the end of the section, explain how to it extends to more general entangled states.

\subsection{The resource space $\mathcal R$ and shift operations on this space}

The resource space is the Hilbert space $\ell^2$, whose orthonormal basis is countably infinite.
In order to define the operations used in the protocol, it is useful to think of this space in terms of countably infinite tensor products of states, where all but finitely many of them are fixed.

First, consider the set of infinite tensor products of 2-qubit computational basis states, where all but finitely many of them are in state $\ket{00}$.
We can express these states as $\ket{\dots x_2 x_1 x_0\,,\, \dots y_2 y_1 y_0}$, 
or as $\ket{x,y}$ (where $x, y \in \NN$, and $x_j$ and $y_j$ are the binary digits of $x$ and $y$, respectively, in position $j$).

Next, consider the set of infinite tensor products of 2-qubit Bell basis states where all but finitely many of them are $\frac{1}{\sqrt 2}\ket{00} + \frac{1}{\sqrt 2}\ket{11}$.
Let us denote these states as $\ket{0.x_{-1}x_{-2}x_{-3}\dots\,,\, 0.y_{-1}y_{-2}y_{-3}\dots}$, with the convention that the qubits in position $-j$ (i.e., the $j$\textsuperscript{th} qubit pair, corresponding to bits $x_{-j}$ and $y_{-j}$) are in the Bell basis state  
\begin{align}
\textstyle{\frac{1}{\sqrt 2}}\ket{0y_{-j}} + \textstyle{\frac{1}{\sqrt 2}}(-1)^{x_{-j}}\ket{1\overline{y_{-j}}}.
\end{align}

A convenient way of denoting an orthonormal basis for the (spatial) tensor product of the Hilbert spaces generated by the two aforementioned sets is as the set of all 
\begin{align}\label{eq:basis-state}
\ket{\dots x_2 x_1 x_0\,\mbox{\Large .}\, x_{-1} x_{-2} \dots\,,\ 
\dots y_2 y_1 y_0\,\mbox{\Large .}\, y_{-1} y_{-2}\dots},
\end{align}
with all but finitely many $x_j$ and $y_{j}$ set to 0.
Equivalently, each basis state can be written as $\ket{x,y}$, where $x$ and $y$ are dyadic rational numbers%
\footnote{A dyadic rational number is of the form 
$x = a/2^b$, where $a, b \in \NN$.
Each dyadic $x$ can be written in binary as $x = x_{\ell} \dots x_2 x_1 x_0 \hspace*{0.3mm}\mbox{\Large .}\hspace*{0.3mm} x_{-1} x_{-2} \dots x_{-r}$.
Formally, for all $j \in \ZZ$, \textit{bit $j$ of $x$} is defined as $x_j = \lfloor x 2^{-j} \rfloor \bmod 2$.
Note that $2x$ is $x$ with all the binary digits shifted left by 1.}.

Intuitively, these basis states can be thought of as two-way infinite tensor products, as illustrated in Figure~\ref{fig:tensor-structure}.
\begin{figure}[!ht]
\begin{center}
\setlength{\unitlength}{0.14mm}
\begin{picture}(40,160)(0,0)
\put(0,70){\makebox(40,40){\small $\cdots$}}
\put(12,0){\makebox(20,20){\tiny $\cdots$}}
\end{picture}
\begin{picture}(400,160)(0,0)
\linethickness{0.5pt}
\put(20,40){\circle{20}}
\put(20,140){\circle{20}}
\put(80,40){\circle{20}}
\put(80,140){\circle{20}}
\put(140,40){\circle{20}}
\put(140,140){\circle{20}}
\put(200,40){\circle{20}}
\put(200,140){\circle{20}}
\put(260,40){\circle{20}}
\put(260,140){\circle{20}}
\put(320,40){\circle{20}}
\put(320,140){\circle{20}}
\put(380,40){\circle{20}}
\put(380,140){\circle{20}}
\put(260,50){\line(0,1){80}}
\put(320,50){\line(0,1){80}}
\put(380,50){\line(0,1){80}}
\put(10,0){\makebox(20,20){\tiny $3$}}
\put(70,0){\makebox(20,20){\tiny $2$}}
\put(130,0){\makebox(20,20){\tiny $1$}}
\put(190,0){\makebox(20,20){\tiny $0$}}
\put(245,0){\makebox(20,20){\tiny $-1$}}
\put(305,0){\makebox(20,20){\tiny $-2$}}
\put(365,0){\makebox(20,20){\tiny $-3$}}
\end{picture}
\begin{picture}(40,160)(0,0)
\put(0,70){\makebox(40,40){\small $\cdots$}}
\put(6,0){\makebox(20,20){\tiny $\cdots$}}
\end{picture}
\end{center}
\caption{\small Schematic picture of the tensor product structure of the basis states. Each circle represents a qubit. In positions $0, 1, 2, \dots$ the qubits are in computational basis states. In positions $-1, -2, \dots$ the qubits are in Bell basis states.}\label{fig:tensor-structure}
\end{figure}
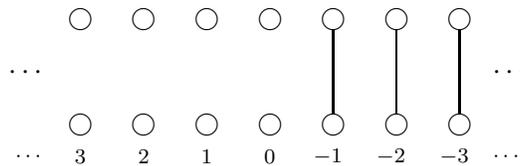
On the left are computational basis states (with all but finitely many in state $\ket{00}$).
On the right are Bell basis states (with all but finitely many
in state $\frac{1}{\sqrt 2}\ket{00} + \frac{1}{\sqrt 2}\ket{11})$.

We now define a \textit{left shift} $L$ on the Hilbert space spanned by these basis states.
Intuitively, $L$ shifts the two-way infinite tensor product to the left by one.
Formally, we define $L$ as the product of two unitaries.
First, define $L_1$ as the left shift of the digits of $x$ and $y$
\begin{align}
L_1 & \ket{\dots x_2 x_1 x_0 \,.\, x_{-1} x_{-2} \dots\,,
\ \dots y_2 y_1 y_0\,.\, y_{-1} y_{-2} \dots} \\
&= 
\ket{\dots x_1 x_0 x_{-1}\,.\, x_{-2} y_{-3} \dots\,,
\ \dots y_1 y_0 y_{-1} \,.\, y_{-2} y_{-3} \dots},\end{align}
or, equivalently, as
$L_1 \ket{x,y} = \ket{2x,2y}$.
$L_1$ is unitary because it is a permutation of the basis states.
Note that $L_1$ does not implement the desired left shift $L$ because the qubits in position $-1$ are in the Bell basis whereas the qubits in position $0$ are in the computational basis.
A basis conversion is needed when position $-1$ is shifted to position $0$.
Define $L_2$ to perform this basis conversion in position $0$ as
\begin{align}
L_2 & \ket{\dots x_2 x_1 x_0 \,.\, x_{-1} x_{-2} \dots \,,
\ \dots y_2 y_1 y_0 \,.\, y_{-1} y_{-2}\dots} \\
&= 
\sqq\ket{\dots x_2 x_1 0 \,.\, x_{-1} x_{-2} \dots \,,
\ \dots y_2 y_1 y_0 \,.\, y_{-1} y_{-2}\dots} \\
&\ \ + \sqq(-1)^{x_0}\ket{\dots x_2 x_1 1 \,.\, x_{-1} x_{-2} \dots \,,
\ \dots y_2 y_1 \overline{y_0} \,.\, y_{-1} y_{-2}\dots}.
\end{align}
This can be equivalently expressed in terms of arithmetic operations on dyadic rationals as
\begin{align}
L_2 \ket{x,y} 
= 
\sqq\ket{x - x_0,y} + \sqq(-1)^{x_0}\ket{x - x_0 + 1, y - 2y_0 + 1}.
\end{align}
$L_2$ is unitary because it is a direct sum of $4 \times 4$ unitaries.
Finally, define $L = L_2 L_1$, which is unitary because $L_1$ and $L_2$ are unitary.

An interesting property  of $L$ is that applying this operation to the state $\ket{0.0,0.0}$ yields $\sqq\ket{0.0,0.0}+ \sqq\ket{1.0,1.0}$.
In the tensor product picture, $L$ leaves the state of all the qubits intact except for the qubits in position 0, whose state changes from $\ket{00}$ to 
$\frac{1}{\sqrt 2}\ket{00} + \frac{1}{\sqrt 2}\ket{11}$.
This is performing something like an embezzlement transformation (in a manner reminiscent of the imaginary ``Hilbert hotel"); 
however, this $L$ does not decompose into two commuting operations that have the structure illustrated in Figure~\ref{fig:commuting-structure}.
In order to obtain such a decomposition, we need to enlarge our Hilbert space.

We begin with some intuition.
In Figure~\ref{fig:tensor-structure}, assume that Alice possesses the qubits in the first row and Bob possesses the qubits in the second row. 
When Alice's qubits are shifted to the left by one, the picture changes to that of Figure~\ref{fig:tensor-offset}.
\begin{figure}[!ht]
\begin{center}
\setlength{\unitlength}{0.14mm}
\begin{picture}(40,180)(0,0)
\put(0,70){\makebox(40,40){\small $\cdots$}}
\put(12,0){\makebox(20,20){\tiny $\cdots$}}
\put(12,160){\makebox(20,20){\tiny $\cdots$}}
\end{picture}
\begin{picture}(400,180)(0,0)
\linethickness{0.5pt}
\put(20,40){\circle{20}}
\put(20,140){\circle{20}}
\put(80,40){\circle{20}}
\put(80,140){\circle{20}}
\put(140,40){\circle{20}}
\put(140,140){\circle{20}}
\put(200,40){\circle{20}}
\put(200,140){\circle{20}}
\put(260,40){\circle{20}}
\put(260,140){\circle{20}}
\put(320,40){\circle{20}}
\put(320,140){\circle{20}}
\put(380,40){\circle{20}}
\put(380,140){\circle{20}}
\put(206,131){\line(3,-5){49}}
\put(266,131){\line(3,-5){49}}
\put(326,131){\line(3,-5){49}}
\put(386,131){\line(3,-5){16}}
\put(10,0){\makebox(20,20){\tiny $3$}}
\put(70,0){\makebox(20,20){\tiny $2$}}
\put(130,0){\makebox(20,20){\tiny $1$}}
\put(190,0){\makebox(20,20){\tiny $0$}}
\put(245,0){\makebox(20,20){\tiny $-1$}}
\put(305,0){\makebox(20,20){\tiny $-2$}}
\put(365,0){\makebox(20,20){\tiny $-3$}}
\put(10,160){\makebox(20,20){\tiny $3$}}
\put(70,160){\makebox(20,20){\tiny $2$}}
\put(130,160){\makebox(20,20){\tiny $1$}}
\put(190,160){\makebox(20,20){\tiny $0$}}
\put(245,160){\makebox(20,20){\tiny $-1$}}
\put(305,160){\makebox(20,20){\tiny $-2$}}
\put(365,160){\makebox(20,20){\tiny $-3$}}\end{picture}
\begin{picture}(5,180)(0,0)
\end{picture}
\begin{picture}(40,180)(0,0)
\put(0,70){\makebox(40,40){\small $\cdots$}}
\put(6,0){\makebox(20,20){\tiny $\cdots$}}
\put(6,160){\makebox(20,20){\tiny $\cdots$}}
\end{picture}
\end{center}
\caption{\small Schematic picture of the tensor product of basis states when Alice's qubits (top row) are shifted left by one.
Pairs of circles connected by lines are in the Bell basis.
Such states are orthogonal to all states of the form of Figure~\ref{fig:tensor-structure}.}
\label{fig:tensor-offset}
\end{figure}
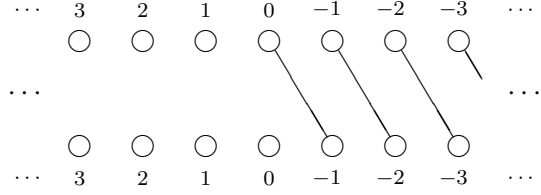
This can be equivalently expressed by shifting the labels of Alice's qubits as in Figure~\ref{fig:tensor-offset-equiv}.
\begin{figure}[!ht]\label{fig:tensor-offset2}
\begin{center}
\setlength{\unitlength}{0.14mm}
\begin{picture}(40,180)(0,0)
\put(0,70){\makebox(40,40){\small $\cdots$}}
\put(12,0){\makebox(20,20){\tiny $\cdots$}}
\put(12,160){\makebox(20,20){\tiny $\cdots$}}
\end{picture}
\begin{picture}(400,180)(0,0)
\linethickness{0.5pt}
\put(20,40){\circle{20}}
\put(20,140){\circle{20}}
\put(80,40){\circle{20}}
\put(80,140){\circle{20}}
\put(140,40){\circle{20}}
\put(140,140){\circle{20}}
\put(200,40){\circle{20}}
\put(200,140){\circle{20}}
\put(260,40){\circle{20}}
\put(260,140){\circle{20}}
\put(320,40){\circle{20}}
\put(320,140){\circle{20}}
\put(380,40){\circle{20}}
\put(380,140){\circle{20}}
\put(260,50){\line(0,1){80}}
\put(320,50){\line(0,1){80}}
\put(380,50){\line(0,1){80}}
\put(10,0){\makebox(20,20){\tiny $3$}}
\put(70,0){\makebox(20,20){\tiny $2$}}
\put(130,0){\makebox(20,20){\tiny $1$}}
\put(190,0){\makebox(20,20){\tiny $0$}}
\put(245,0){\makebox(20,20){\tiny $-1$}}
\put(305,0){\makebox(20,20){\tiny $-2$}}
\put(365,0){\makebox(20,20){\tiny $-3$}}
\put(10,160){\makebox(20,20){\tiny $4$}}
\put(70,160){\makebox(20,20){\tiny $3$}}
\put(130,160){\makebox(20,20){\tiny $2$}}
\put(190,160){\makebox(20,20){\tiny $1$}}
\put(250,160){\makebox(20,20){\tiny $0$}}
\put(305,160){\makebox(20,20){\tiny $-1$}}
\put(365,160){\makebox(20,20){\tiny $-2$}}
\end{picture}
\begin{picture}(40,180)(0,0)
\put(0,70){\makebox(40,40){\small $\cdots$}}
\put(6,0){\makebox(20,20){\tiny $\cdots$}}
\put(6,160){\makebox(20,20){\tiny $\cdots$}}
\end{picture}
\end{center}
\caption{\small An alternative schematic picture of the tensor product of Figure~\ref{fig:tensor-offset}, where the labels of the qubits on the top row are adjusted to reflect the shift in the top row.}\label{fig:tensor-offset-equiv}
\end{figure}
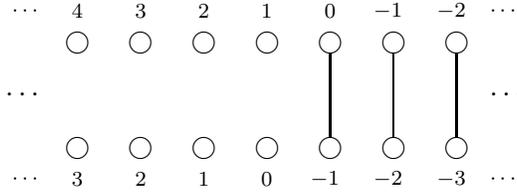
More generally, for an arbitrary $r \in \ZZ$, a left shift of the first row by $r$ is illustrated in Figure~\ref{fig:tensor-offset-r}.
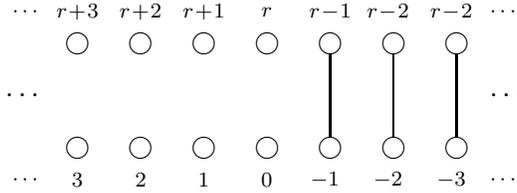
\begin{figure}[!ht]\label{fig:tensor-offset3}
\begin{center}
\setlength{\unitlength}{0.14mm}
\begin{picture}(40,180)(0,0)
\put(0,70){\makebox(40,40){\small $\cdots$}}
\put(12,0){\makebox(20,20){\tiny $\cdots$}}
\put(12,160){\makebox(20,20){\tiny $\cdots$}}
\end{picture}
\begin{picture}(400,180)(0,0)
\linethickness{0.5pt}
\put(20,40){\circle{20}}
\put(20,140){\circle{20}}
\put(80,40){\circle{20}}
\put(80,140){\circle{20}}
\put(140,40){\circle{20}}
\put(140,140){\circle{20}}
\put(200,40){\circle{20}}
\put(200,140){\circle{20}}
\put(260,40){\circle{20}}
\put(260,140){\circle{20}}
\put(320,40){\circle{20}}
\put(320,140){\circle{20}}
\put(380,40){\circle{20}}
\put(380,140){\circle{20}}
\put(260,50){\line(0,1){80}}
\put(320,50){\line(0,1){80}}
\put(380,50){\line(0,1){80}}
\put(10,0){\makebox(20,20){\tiny $3$}}
\put(70,0){\makebox(20,20){\tiny $2$}}
\put(130,0){\makebox(20,20){\tiny $1$}}
\put(190,0){\makebox(20,20){\tiny $0$}}
\put(245,0){\makebox(20,20){\tiny $-1$}}
\put(305,0){\makebox(20,20){\tiny $-2$}}
\put(365,0){\makebox(20,20){\tiny $-3$}}
\put(10,160){\makebox(20,20){\tiny $r\!+\!3$}}
\put(70,160){\makebox(20,20){\tiny $r\!+\!2$}}
\put(130,160){\makebox(20,20){\tiny $r\!+\!1$}}
\put(190,160){\makebox(20,20){\tiny $r$}}
\put(250,160){\makebox(20,20){\tiny $r\!-\!1$}}
\put(305,160){\makebox(20,20){\tiny $r\!-\!2$}}
\put(365,160){\makebox(20,20){\tiny $r\!-\!2$}}
\end{picture}
\begin{picture}(40,180)(0,0)
\put(0,70){\makebox(40,40){\small $\cdots$}}
\put(6,0){\makebox(20,20){\tiny $\cdots$}}
\put(6,160){\makebox(20,20){\tiny $\cdots$}}
\end{picture}
\end{center}
\caption{\small Schematic picture of a left shift of Alice's qubits (top row) by $r \in \ZZ$.}\label{fig:tensor-offset-r}
\end{figure}

With the picture of Figure~\ref{fig:tensor-offset-r} in mind, define the Hilbert space $\mathcal R$ as having orthonormal basis states of the form $\ket{r,x,y}$, where $x, y$ are dyadic rationals and $r \in \ZZ$ represents the leftward shift of Alice's qubits.
We can interpret $\ket{r,x,y}$ as an encoding of the following \textit{logical state}.
For all $j \in \ZZ$, \textit{Alice's logical qubit in position $j+r$} and \textit{Bob's logical qubit in position $j$} are in the joint state
\begin{align}
\begin{cases}
\ket{x_{j}y_{j}} & \text{if $j \ge 0$} \\[2mm]
\sqq\ket{0y_{j}}+\sqq(-1)^{x_{j}}\ket{1\overline{y_{j}}} & \text{if $j < 0$.}
\end{cases}
\end{align}

Now we define the \textit{Alice left shift} $L_A$ as simply
\begin{align}
L_A \ket{r,x,y} = \ket{r+1,x,y}.
\end{align}
$L_A$ is obviously unitary and commutes with $L$, since they act on different 
components of $\ket{r,x,y}$.

Next, define the \textit{Bob left shift} $L_B$ as
\begin{align}
L_B = L_{A}^{*}L
\end{align}
(a left shift of both Alice and Bob's qubits followed by a right shift of Alice's qubits).
Note that $L_A$ and $L_B$ commute, since $L_A$ and $L$ commute.
Also, since $L$ is a left shift by both Alice and Bob and $L_{A}^{*}$ is a right shift by Alice, $L_B$ has no net effect on Alice's logical qubits.

\subsection{Swap operations between $\HH_A$, $\HH_B$ and $\mathcal R$}
Prior to defining our embezzlement protocol, we define 
swap operations between $\HH_A$ and the logical qubit of Alice in position 0 of $\mathcal R$, and between $\HH_B$ and the logical qubit of Bob in position 0 of $\mathcal R$.

The \textit{Bob swap} $S_B$ is defined simply as the unitary operation that acts on $\HH_B \otimes \mathcal R$
as
\begin{align}
S_B \ket{t} \otimes \ket{r,\, x,\, \dots y_1 y_0 \,.\, y_{-1} \dots}
= 
\ket{y_0} \otimes \ket{r,\, x,\, \dots y_1 t \,.\, y_{-1} \dots},
\end{align}
or, equivalently, as
$S_B\ket{t}\otimes\ket{r,x,y} = \ket{y_0}\otimes\ket{r,x,y-y_0+t}$.
$S_B$ is clearly unitary and commutes with $L_A$ since they act on different components of  each basis state $\ket{r,x,y} \in \mathcal R$.

The corresponding \textit{Alice swap}, acting on $\HH_A \otimes \mathcal R$, is more complicated than $S_B$.
First define $\tilde{S}_A$ (a \textit{na\"ive Alice swap}) as
\begin{align}
\tilde{S}_A \ket{s} \otimes \ket{r,\, \dots x_1 x_0 \,.\, x_{-1} \dots,\, y}
= 
\ket{x_0} \otimes \ket{r,\, \dots x_1 s \,.\, x_{-1} \dots,\, y},
\end{align}
or, equivalently, as
$\tilde{S}_A\ket{s}\otimes\ket{r,x,y} = \ket{x_0}\otimes \ket{r,x-x_0+s,y}$.
$\tilde{S}_A$ does not swap with Alice's logical qubit in position 0---moreover, $\tilde{S}_A$ does not commute with $L_B$.
To swap with Alice's logical qubit in position 0, it is convenient to first define the \textit{controlled-$L$}, denoted as $C$,
acting on $\mathcal R$ as 
\begin{align}
C \ket{r,x,y} = L^{r}\ket{r,x,y},
\end{align}
which makes sense because $L^{r}$ acts only on the second and third component of 
$\ket{r,x,y}$.
$C$ is unitary because each $L^{r}$ is unitary and $C$ is a direct sum of all~$L^{r}$.
Intuitively, $C^*\ket{r,x,y}$ is a state in which Alice's \textit{literal} qubit in position $0$ corresponds to Alice's logical qubit in position 0 in $\ket{r,x,y}$.
Now we define the actual \textit{Alice swap} as
\begin{align}
S_A = C \tilde{S}_A C^{*}.
\end{align}
Clearly $S_A$ is unitary and, for each $\ket{r,x,y} \in \mathcal R$, its effect is localized to Alice's logical qubit in position $0$.
$S_A$ and $S_B$ commute because $S_B$ is localized to Bob's logical qubit in position $0$.
Moreover, $S_A$ and $L_B$ commute because, for $\ket{r,x,y} \in \mathcal R$, $L_B$ is localized to Bob's logical qubits.

\subsection{The embezzlement protocol}

The idea is to start with state $\ket{0,\,0.0,\,0.0}$, perform $L_A$ and $L_B$, and then swap the two qubits in position 0 of $\mathcal R$ into $\HH_A$ and $\HH_B$.

Alice performs $U_A = S_A L_A$ and Bob performs $U_B = S_B L_B$. 
Clearly $U_A$ and $U_B$ commute.
Recall that $L_A L_B = L$. The state evolves as:
\medskip
\begin{description}
\item[0. initial state]
$\ket{0}\otimes\ket{0}\otimes\ket{0,\,0.0,\,0.0}$ 
\vspace*{2mm}
\item[1. after $L_A L_B$]
$\ket{0}\otimes\ket{0}\otimes\bigl(\sqq\ket{0,\,0.0,\,0.0}+\sqq\ket{0,\,1.0,\,1.0}\bigr)$
\item[2. after $S_A S_B$]
\vspace*{2mm}
$\bigl(\sqq\ket{0}\otimes\ket{0}+\sqq\ket{1}\otimes\ket{1}\bigr)\otimes\ket{0,\,0.0,\,0.0}$ \end{description}
\medskip
This completes the protocol for perfect embezzlement in the commuting operator framework.

\begin{remark}
It is easy to adapt the above method to embezzle a more general entangled state, say, of the form $\phi = \sum_{i,j=0}^{d-1} \alpha_{ij} \ket{i} \otimes \ket{j}$.
First, redefine $\mathcal{R}$ to be in terms of basis states of the form $\ket{r,x,y}$, where $x$ and $y$ are $d$-adic rational numbers (i.e., with digits in $\mathbb{Z}_d$).
Then it suffices to set the operation $L_2$ (which is the basis change part of the left shift operation) to be any unitary operation on $\mathbb{C}^d \otimes \mathbb{C}^d$ that maps $\ket{0}\otimes\ket{0}$ to $\phi$.
The other parts of the protocol are essentially the same.
\end{remark}


\section{Coherent embezzlement games}

A purported protocol for embezzlement cannot be tested in the way that nonlocal games can, because 
Alice and Bob can perform local operations that perfectly map $\ket{0}\otimes\ket{0}$ to 
$\frac{1}{\sqrt 2}\ket{0}\otimes\ket{0}+\frac{1}{\sqrt 2}\ket{1}\otimes\ket{1}$ 
using the resource of only a single (concealed) Bell state.
Leung, Toner and Watrous~\cite{LeungTW13} proposed a \textit{coherent state exchange} game that is related to embezzlement but is operationally testable.
In this game, Alice and Bob each receive a qutrit from a referee as input and they each return a qubit to the referee, who performs a measurement on the returned state to determine whether they win or lose.
There is no perfect strategy for this game using finite entanglement.
It is shown in~\cite{LeungTW13} that, for all $\epsilon > 0$: there exists a strategy that succeeds with probability $1 - \epsilon$ using $O(\log(1/\epsilon))$-entropy entanglement; moreover, to succeed with probability $1 - \epsilon$ requires entanglement with entropy $\Omega(\log(1/\epsilon))$.

Regev and Vidick~\cite{RV2013} presented a simplification of the coherent state exchange game that has the above properties, but where the outputs are classical bits instead of qubits (the inputs are still qutrits).
In~\cite{RV2013}, this is called the $T_2$ game.
We refer to this as the \textit{coherent embezzlement} game, to highlight its close relationship with embezzlement.

In this section, we begin by reviewing the definition of the coherent embezzlement (a.k.a.\ $T_2$) game.
Then we show that a perfect strategy for embezzlement can be converted to a perfect strategy for coherent embezzlement---and vice versa.
By such reductions, we prove that there is a perfect strategy for coherent embezzlement in the commuting operator framework (Theorem~\ref{thm:perfect-T2}), but there is no such perfect strategy in the tensor product framework (Theorem~\ref{thm:no-perfect-T2}).

We now define the coherent embezzlement game~\cite{RV2013}.
Alice and Bob each receive two qutrits as input
and they each produce a classical bit as output.
The input state that Alice and Bob jointly receive is either $\phi_0$ or $\phi_1$, where
\begin{align}
\phi_0 &= \sqq\ket{0}\otimes\ket{0}+\sqq\bigl(\sqq\ket{1}\otimes\ket{1}+\sqq\ket{2}\otimes\ket{2}\bigr) \\
\phi_1 &= \sqq\ket{0}\otimes\ket{0}-\sqq\bigl(\sqq\ket{1}\otimes\ket{1}+\sqq\ket{2}\otimes\ket{2}\bigr).
\end{align}
Call Alice and Bob's output bits $a$ and $b$ respectively.
The winning condition is that: when the input is $\phi_0$, $a \oplus b = 0$; when the input is $\phi_1$, $a \oplus b = 1$.

Note that the winning condition does not require Alice and Bob's resource state to be used in a catalytic manner; Alice and Bob are free to destroy this state in their strategy.
In this sense, the coherent embezzlement game is simpler than embezzlement, whose definition depends critically on restoring the resource state.

In the remainder of this section, for technical convenience, we represent qutrits as pairs of qubits using the encoding $0 \equiv 00$, $1 \equiv 10$ (for Alice) or $01$ (for Bob), and $2 \equiv 11$.
With this encoding, the input states that Alice and Bob jointly receive can be written as, for $c \in \{0,1\}$,
\begin{align}
\phi_c &= \sqq\ket{00}\otimes\ket{00}+\sqq(-1)^c\bigl(\sqq\ket{10}\otimes\ket{01}+\sqq\ket{11}\otimes\ket{11}\bigr) \\
&= \sqq\ket{0}\ket{00}\ket{0}+\sqq(-1)^c\ket{1}\bigl(\sqq\ket{00}+\sqq\ket{11}\bigr)\ket{1}
\end{align}
(where the first two qubits are Alice's input and the last two qubits are Bob's input).

\begin{thm}\label{thm:perfect-T2}
There is a perfect strategy for the coherent embezzlement game in the commuting operator framework.
\end{thm}

\begin{proof}
Let $\mathcal{R}$, $U_A$, and $U_B$ be as defined in the unitary embezzlement protocol of section~\ref{sec:explicit}.
For coherent embezzlement, the input Hilbert spaces are 
$\mathcal{H}_{A_1} \otimes \mathcal{H}_{A_2} = \mathcal{H}_{B_1} \otimes \mathcal{H}_{B_2} = \mathbb{C}^2\otimes\mathbb{C}^2$.
We will show that the protocol in Figure~\ref{fig:embezzlement-to-coherent} performs coherent embezzlement.
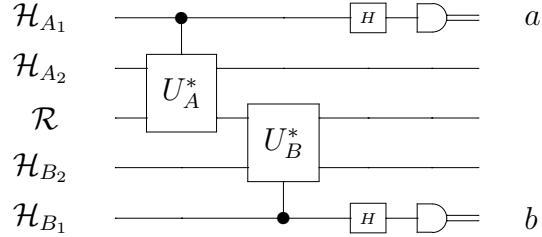
\begin{figure}[!ht]
\centerline{
\Qcircuit @C=1em @R=.7em{
\lstick{\mathcal{H}_{A_1}\ \ \ } & \ctrl{1} & \qw & \gate{\scriptscriptstyle H} & \measureD{\!\!\scriptscriptstyle\frac{}{}} & \cw & \rstick{a} \\
\lstick{\mathcal{H}_{A_2}\ \ \ } & \multigate{1}{U^{*}_{A}}        & \qw          & \qw & \qw & \qw \\
\lstick{\mathcal{R} \ \ \ \ } & \ghost{U^{*}_{A}}        & \multigate{1}{U^{*}_{B}}          & \qw & \qw &\qw \\
\lstick{\mathcal{H}_{B_2}\ \ \ } & \qw & \ghost{U^{*}_{B}}    & \qw  & \qw & \qw \\
\lstick{\mathcal{H}_{B_1}\ \ \ } & \qw & \ctrl{-1} & \gate{\scriptscriptstyle H} & \measureD{\!\!\scriptscriptstyle\frac{}{}} & \cw & \rstick{b}
}
}

\caption{\small Circuit diagram of a protocol for coherent embezzlement from a protocol for embezzlement based on $U_A$, $U_B$, and $\psi \in \mathcal R$.}
\label{fig:embezzlement-to-coherent}
\end{figure}
First note that, since the controlled $U^*_A$ and $U^*_B$ perform the inverse of embezzlement when their control qubits are in state $\ket{1}$, they perform a mapping on $\mathcal{H}_{A_1} \otimes \mathcal{H}_{A_2} \otimes \mathcal{H}_{B_2} \otimes \mathcal{H}_{B_1}$ such that
\begin{align}
\phi_c =&\sqq\ket{0}\ket{00}\ket{0}+\sqq(-1)^c\ket{1}\bigl(\sqq\ket{00}+\sqq\ket{11}\bigr)\ket{1} \\[1mm]
&\mapsto \ \ \sqq\ket{0}\ket{00}\ket{0}+\sqq(-1)^c\ket{1}\ket{00}\ket{1}.
\end{align}
Note that, on the Hilbert space $\mathcal{H}_{A_1} \otimes \mathcal{H}_{B_1}$, this is the pure state $\sqq\ket{00}+\sqq(-1)^c\ket{11}$.
Finally, since the Hadamard gates perform the mapping 
\begin{align}
\sqq\ket{00}+\sqq(-1)^c\ket{11}
\mapsto 
\sqq\ket{0c}+\sqq\ket{1\overline{c}},
\end{align}
the result follows.
\end{proof}

\begin{thm}\label{thm:no-perfect-T2}
There is no perfect strategy for the coherent embezzlement game in the tensor product framework (where Alice and Bob are allowed to use ancillas).
\end{thm}

\begin{proof}
The idea is that a perfect strategy for coherent embezzlement can be converted into a perfect strategy for embezzlement.
Suppose that there is a perfect strategy for coherent embezzlement.
Without loss of generality, it can be assumed that this strategy is of the form depicted in Figure~\ref{fig:wlog-coherent}, where $\psi \in \mathcal{R}_A \otimes \mathcal{R}_B$, Alice and Bob each perform a local unitary operation followed by measurement of one of their qubits.
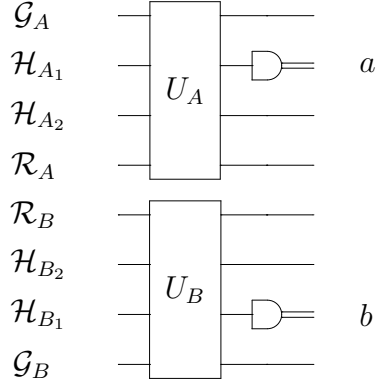
\begin{figure}[!ht]
\centerline{
\Qcircuit @C=1em @R=.7em{
\lstick{\mathcal{G}_A \ \ \ \ \ } & \multigate{3}{U_A} & \qw & \qw \\
\lstick{\mathcal{H}_{A_1} \ \ \ \, } & \ghost{U_{A}} & \measureD{\!\!\scriptscriptstyle\frac{}{}} & \cw & \rstick{a} \\
\lstick{\mathcal{H}_{A_2} \ \ \ \, } & \ghost{U_{A}}        & \qw          & \qw \\
\lstick{\mathcal{R}_A \ \ \ \ \, } & \ghost{U_{A}}        & \qw          & \qw \\
\lstick{\mathcal{R}_B \ \ \ \ \, }  & \multigate{3}{U_{B}} &  \qw         & \qw \\
\lstick{\mathcal{H}_{B_2} \ \ \ \, } & \ghost{U_{B}}        & \qw & \qw  \\
\lstick{\mathcal{H}_{B_1} \ \ \ \, }  & \ghost{U_{B}}        & \measureD{\!\!\scriptscriptstyle\frac{}{}} & \cw & \rstick{b} \\
\lstick{\mathcal{G}_B \ \ \ \ \ } & \ghost{U_B} & \qw & \qw 
}
}
\caption{\small Circuit diagram of an arbitrary protocol for coherent embezzlement.
$\psi \in \mathcal{R}_A \otimes \mathcal{R}_B$ is the resource state.
Without loss of generality, two local unitaries are performed, followed by measurements of two specific qubits.}\label{fig:wlog-coherent}
\end{figure}
Since we are assuming that the strategy is perfect, for input state $\phi_0$, the output bits satisfy $a\oplus b = 0$, and for input state $\phi_1$, the output bits satisfy $a\oplus b = 1$.

Now, consider the protocol in Figure~\ref{fig:coherent-to-embezzlement}.
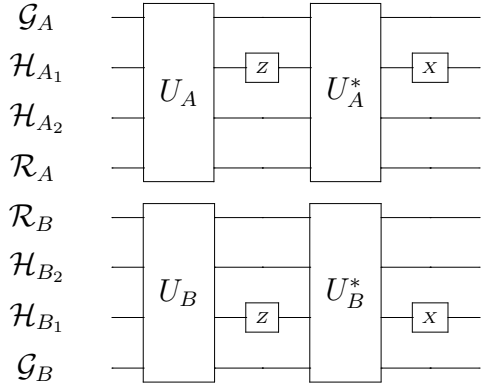
\begin{figure}[!ht]
\centerline{\hspace*{15mm}
\Qcircuit @C=1em @R=.7em{
\lstick{\mathcal{G}_A \ \ \ \ } & \multigate{3}{U_A} & \qw & \multigate{3}{U^{*}_{A}} & \qw & \qw \\
\lstick{\mathcal{H}_{A_1} \ \ \ } & \ghost{U_A} & \gate{\scriptscriptstyle Z} & \ghost{U^{*}_{A}} & \gate{\scriptscriptstyle X} & \qw \\
\lstick{\mathcal{H}_{A_2} \ \ \ } & \ghost{U_A}        & \qw           & \ghost{U^{*}_{A}} & \qw & \qw \\
\lstick{\mathcal{R}_A \ \ \ \ }& \ghost{U_A}  & \qw & \ghost{U^{*}_{A}} & \qw & \qw \\
\lstick{\mathcal{R}_B \ \ \ \ } & \multigate{3}{U_B} &  \qw         & \multigate{3}{U^{*}_{B}} & \qw & \qw \\
\lstick{\mathcal{H}_{B_2} \ \ \ }& \ghost{U_B}        & \qw & \ghost{U^{*}_{B}} & \qw & \qw \\
\lstick{\mathcal{H}_{B_1} \ \ \ } & \ghost{U_B}        & \gate{\scriptscriptstyle Z} & \ghost{U^{*}_{B}} & \gate{\scriptscriptstyle X} & \qw \\
\lstick{\mathcal{G}_B \ \ \ \ } & \ghost{U_B} & \qw & \ghost{U^{*}_{B}} & \qw & \qw}
}
\caption{\small Modification of circuit in Figure~\ref{fig:wlog-coherent} that performs embezzlement using the resource state $\psi \in \mathcal{R}_A \otimes \mathcal{R}_B$.}
\label{fig:coherent-to-embezzlement}
\end{figure}
We begin by showing that, for each $c \in \{0,1\}$, the effect of the first three steps of the protocol in Figure~\ref{fig:coherent-to-embezzlement} is to map $\phi_c\otimes\psi\otimes\gamma_A\otimes\gamma_B$ to $(-1)^c \phi_c\otimes\psi\otimes\gamma_A\otimes\gamma_B$.
Since $a \oplus b = c$, 
After the $U_A \otimes U_B$ operations have been performed, the state must be of the form
\begin{align}
\alpha\ket{0}\otimes\ket{c}\otimes\xi_0 + \beta\ket{1}\otimes\ket{\overline{c}}\otimes\xi_1,
\end{align}
when expressed the state in the Hilbert space 
\begin{align}
\mathcal{H}_{A_1} \otimes \mathcal{H}_{B_1} \otimes (\mathcal{H}_{A_2} \otimes \mathcal{H}_{B_2} \otimes \mathcal{R}_A \otimes \mathcal{R}_B \otimes \mathcal{G}_A \otimes \mathcal{G}_B).
\end{align}
Therefore, after the two $Z$ gates, the state is 
\begin{align}
(-1)^c\bigl(\alpha\ket{0}\otimes\ket{c}\otimes\xi_0 + \beta\ket{1}\otimes\ket{\overline{c}}\otimes\xi_1\bigr),
\end{align}
and this is mapped by $U^{*}_{A} \otimes U^{*}_{B}$ to $(-1)^c \phi_c\otimes\psi\otimes\gamma_A\otimes\gamma_B$.

This implies that the first three steps of the modified protocol maps 
\begin{align}
\ket{00}\otimes\ket{00} = \sqq\phi_0 + \sqq\phi_1
\end{align}
on $\mathcal{H}_{A}\otimes\mathcal{H}_B$ (with the other registers in state $\psi\otimes\gamma_A\otimes\gamma_B$) to
\begin{align}
\sqq\ket{10}\otimes\ket{01}+\sqq\ket{11}\otimes\ket{11} = \sqq\phi_0 - \sqq\phi_1.
\end{align}

Finally, the two $X$ gates of the protocol in Figure~\ref{fig:coherent-to-embezzlement} map the state 
$\sqq\ket{10}\otimes\ket{01}+\sqq\ket{11}\otimes\ket{11}$
to
$\sqq\ket{00}\otimes\ket{00}+\sqq\ket{01}\otimes\ket{10}$.
Therefore, the entire protocol   
maps
$\ket{00}\otimes\ket{00}\otimes\psi\otimes\gamma_A\otimes\gamma_B$
to 
\begin{align}
  \bigl(\sqq\ket{00}\otimes\ket{00}&+\sqq\ket{01}\otimes\ket{10}\bigr)\otimes\psi\otimes\gamma_A\otimes\gamma_B \\
& =\ket{0}(\sqq\ket{00}+\sqq\ket{11})\ket{0}\otimes\psi\otimes\gamma_A\otimes\gamma_B.
\end{align}
This is perfect embezzlement in the registers $\mathcal{H}_{A_2} \otimes \mathcal{H}_{B_2}$.

Since this protocol violates Theorem~\ref{thm:no-spatial-embezzlement}, there cannot exist a perfect strategy for coherent embezzlement using entanglement of the form $\psi \in \mathcal{R}_A \otimes \mathcal{R}_B$.
\end{proof}


\section{Perfect embezzlement requires non-unitary isometries}

In this section we obtain further information about the nature of the unitary operators that appear in perfect embezzlement protocols.
We show that some of the operators that occur in such protocols must contain non-unitary isometries.  This result justifies the need for the register shifts that appear in the explicit protocol of the earlier section.
Also, since non-unitary isometries only exist in infinite dimensions, this result implies that perfect embezzlement is impossible with a finite dimensional resource space.

We begin by reviewing the key facts about unitaries and isometries.  Given a Hilbert space $\cl H$ a linear map $V: \cl H \to \cl H$ is an {\it isometry} provided that $\|Vh\|=\|h\|$ for all $h \in \cl H$. Note that isometries are necessarily one-to-one. The map $V$ is a {\it unitary} provided that it is an isometry and it is onto.

A simple dimension count shows that in finite dimensions every isometry is necessarily a unitary. An example of an isometry that is not a unitary is the {\it unilateral shift S}.  This is the operator on the Hilbert space $\ell^2(\bb N)$ which has an orthonormal basis $\{ \ket{j}: j \in \bb N \}$, defined by $S(\ket{j}) = \ket{j+1}$. It is easily seen that this operator is an isometry and the $\ket{1}$ is orthogonal to the range of $S$, so that $S$ is not onto. Note that the kernel of $S^*$ is equal to the span of $\ket{1}$.

We call a linear map $V: \cl H \to \cl H$ a {\it non-unitary isometry} provided that it is an isometry that is not onto. It is not hard to show that if $V$ is a non-unitary isometry, $e_1$ is any unit vector orthogonal to the range of $V$, and we set $e_n = V^n e_1,$ then this sequence is orthonormal and $V$ acts as a unilateral shift on this subspace.  So, in this sense, non-unitary isometries always contain a space on which they act like unilateral shifts.

Recall that a subspace $\cl M$ is called {\em invariant} for an operator $C$ provided that $C(\cl M) \subseteq \cl M$.

\begin{lemma} Let $\cl H$ be a Hilbert space, $C: \cl H \to \cl H$ be a contraction (i.e., $\|Ch\| \le \|h \|, \, \forall h \in \cl H$), and $\cl M \subseteq \cl H$ be a subspace that is invariant for $C$. If the restriction of $C$ to $\cl M$ acts as a unitary $U$ on $\cl M$, then $\cl M$ is also invariant for $C^*$ and the restriction of $C^*$ to $\cl M$ is $U^*$.
\end{lemma}
\begin{proof} This is easiest to see using operator matrices. Decomposing $\cl H = \cl M \oplus \cl M^{\perp}$ we may write $C= \begin{pmatrix} U &X \\Y & Z \end{pmatrix}$ where $U: \cl M \to \cl M,$ $X: \cl M^{\perp} \to \cl M$, $Y: \cl M \to \cl M^{\perp}$, and $Z: \cl M^{\perp} \to \cl M^{\perp}$.
The fact that $\cl M$ is invariant implies that $Y=0$. The hypothesis that the restriction of $C$ to $\cl M$ is a unitary, means that $U$ is a unitary.

But $C^*= \begin{pmatrix} U^* & X^*\\ 0 & Z^* \end{pmatrix}$ is a contraction with $U^*$ a unitary.  This forces that for every $h \in \cl M,$  $X^*h =0$.  Thus, $X^*=0$ and $\cl M$ is invariant for~$C^*$.
\end{proof}

Note that if $U$ was a non-unitary  isometry, then $U^*$ would have a kernel and we could no longer conclude that $X^*=0$.

The next concept that we shall need is the polar decomposition of an operator. Given an operator $A$ on a Hilbert space $\cl H$, we set $|A|= (A^*A)^{1/2}$.  Let $\cl R(A)$ denote the closure of the range of $A$ and let $\cl R(|A|)$ denote the closure of the range of $|A|$.
Note that for any $x \in \bb \cl H,$ we have
\[ \|Ax\|^2 = \langle Ax, Ax \rangle = \langle x, A^*Ax \rangle =
\langle |A|x, |A|x \rangle = \| |A|x\|^2.\]
From this equality, it follows that there is a well-defined linear isometry $W: \cl R(|A|) \to \cl R(A)$ defined by setting $W(|A|h) = Ah$, so that $A= W|A|.$ If we extend $W$ to a map $U: \cl H \to \cl H$ by sitting $U= WQ$ where $Q$ is the orthogonal projection onto $\cl R(|A|)$ then we still have $A= U|A|$. The map $U$ is called a {\it partial isometry}. This representation $A=U|A|$ is called the {\it polar decomposition} of $A$.

\begin{lemma} 
Let $C$ be a contraction and let $C=U|C|$ be its polar decomposition.  If $h$ is a unit vector such that $||C^nh||=1, \forall n \in \bb N$, then the closed subspace $\cl M$ generated by $\{ C^nh: n \ge 0 \}$ is an invariant subspace for $C$ and for any $v \in \cl M,$  $\|Cv\|=\|v\|$.
\end{lemma}
\begin{proof} Clearly $\cl M$ is invariant for $C$. Set $P = |C|$ so that $C=UP$ with $0 \le P \le I$ and
$||Ph||=||UPh||= ||Ch||=1$. This implies that $Ph=h$.  Thus, $Ch
= Uh$.  Since  $1= ||C^2h|| = ||UPUPh|| = ||UP(Uh)|| = ||P(UH)|| $ we have that $P(Uh) = Uh$ and so $C^2h = UPUPh =UPUh=U^2h.$ Inductively,  $P(U^nh) = U^nh$ so that
$C^nh = U^n h$.  Thus, $Cv= Uv$ for any $v \in \cl M$.

Since $P(U^nh) = U^n h$ for all $n \ge 0$, we have that $\cl M \subseteq \cl R(|C|)$.  But $U$ acts isometrically on $\cl R(|C|)$ and since $U$ and $C$ are equal on this subspace, $C$ acts isometrically on $\cl M$.
\end{proof}

\begin{thm} Let $\cl H_R, \psi, U_A=(U_{ij}),$ and $U_B= (V_{kl})$ be a perfect embezzlement protocol. If  $\cl M$ is the closed subspace of $\cl H$ spanned by $\{ U_{00}^{*n} \psi: n \ge 0 \}$, then the restriction of $U_{00}^*$ to this invariant subspace is a non-unitary isometry.
\end{thm}
\begin{proof}
Recall that the $U_{ij}$'s must *-commute with the $V_{kl}$'s.
The embezzlement relations tell us that
\[ U_{00}V_{00}\psi = \psi/\sqrt{2}, \ \ U_{10}V_{00}\psi =0, \ \ U_{10}V_{10}\psi = \psi/\sqrt{2}, \ \ U_{00}V_{10}\psi
=0.\]

The fact that  $U_A$ and $U_B$ are unitaries implies that
\[ I= V^*_{00}V_{00}+ V^*_{10}V_{10}= U_{00}^*U_{00} + U_{10}^*U_{10} \]
so that
\begin{align*}
U_{00} \psi &= U_{00}(V_{00}^*V_{00} + V_{10}^*V_{10})\psi = V_{00}^*U_{00}V_{00}\psi = V_{00}^*\psi/\sqrt{2}, \\
V_{00}\psi &= V_{00}(U_{00}^*U_{00} +U_{10}^*U_{10})\psi = U_{00}^*\psi/ \sqrt{2}.
\end{align*}
Iterating, yields
\begin{align*}
(V_{00}^*)^n\psi = (\sqrt{2}U_{00})^n \psi \mbox{\ \ and \ \ } V_{00}^n \psi = (U_{00}^*/\sqrt{2})^n \psi.
\end{align*}
Since $(1/\sqrt{2})^n \psi = (U_{00}V_{00})^n\psi = U_{00}^nV_{00}^n \psi =
U_{00}^n(U_{00}^*/\sqrt{2})^n \psi$, we have that $U_{00}^n U_{00}^{*n}\psi =\psi.$
Because $||U_{00}|| \le 1$ we have that $||U_{00}^{*n} \psi|| =1$ for all $n$.

Hence, by the previous lemma, $U_{00}^*$ acts isometrically on $\cl M$.

Now if $U_{00}^*$ acted unitarily, then by the earlier lemma, for vectors in this space $U_{00}$ would be the inverse and in particular would act isometrically on $\cl M$.
But then we would have that $\|V_{00}^*\psi\|= \sqrt{2} \|U_{00}\psi\| = \sqrt{2}$. This is impossible, because $U_B$ is a unitary and so $\|V_{00}^*\| \le 1$.

This contradiction shows that $U_{00}^*$ must be a non-unitary isometry on $\cl M$.
\end{proof}

This yields the following fact. Recall that for a perfect embezzlement protocol, we are only assuming that the operators commute, not that the resource space has a bipartite tensor structure.

\begin{cor} Perfect embezzlement is impossible in the commuting-operator framework if the resource space is finite-dimensional.
\end{cor}
\begin{proof} If $\cl H_R$ is finite dimensional, then $\cl M$ is also finite dimensional. But every isometry on a finite dimensional space is necessarily a unitary, contradicting the fact that $U_{00}^*$ is a non-unitary isometry.
\end{proof}

\section{Acknowledgements}
We would like to thank Marius Junge, Debbie Leung, Volkher Scholz, and John Watrous for helpful discussions. This research was supported in part by Canada's NSERC, a David R. Cheriton Scholarship, and a Mike and Ophelia Lazaridis Fellowship.



\bibliographystyle{amsplain}

\bibliography{InfiniteEntanglement}

%
%
%


\appendix
\section{The Schmidt and polar decompositions in infinite dimensions}\label{appendix:Schmidt}

In this section, for the convenience of the reader, we gather together a few useful results from operator theory that are not well known within the QIT community. We are claiming no originality.

\begin{defn} Let $W: H_1 \to H_2$. Then $W$ is called an {\bf isometry} if $\|Wh_1\|_2 = \|h_1\|_1$ for every $h_1 \in H_1$.  $W$ is called a {\bf coisometry} iff $W^*:H_2 \to H_1$ is an isometry. $W$ is called a {\bf partial isometry} if the restriction of $W$ to $ker(W)^{\perp}$ is an isometry. In this case the space $ker(W)^{\perp}$ is called the {\bf initial space} of $W$ and $ran(W)^- = ker(W^*)^{\perp}$ is called the {\bf final space} of $W$.
\end{defn}

\begin{prop}\label{prop:tensor_basis} Let $H$ and $K$ be Hilbert spaces of arbitrary dimension, let $\{ e_{\alpha}: \alpha \in A \}$ and $\{ f_{\beta} : \beta \in B \}$ be o.n. bases for $H$ and $K$, respectively.  Then $\{ e_{\alpha} \otimes  f_{\beta}: \alpha \in A, \beta \in B \}$ is an o.n. basis for $H \otimes K.$
\end{prop}

\begin{prop}[The Polar Decomposition]\label{prop:polar_decomposition} Let $X: H_1 \to H_2$ and let $|X|= (X^*X)^{1/2}.$
Then there is a unique partial isometry $W:H_1 \to H_2$ with initial space $ker(X)^{\perp}=ker(|X|)^{\perp}= ran(|X|)^-$ and final space $ran(X)^-= ker(X^*)^{\perp}$ such that $X= W|X|.$
\end{prop}

To prove, one simply sets $W(|X|h) = Xh$ and shows that this is well-defined and satisfies the properties.

Note that $W$ is an isometry iff $ker(X) = (0)$ and is a coisometry iff $ran(X)^- = H_2.$

\begin{prop}\label{prop:coisometry_on} Let $\{ e_k \}$ be an o.n. sequence in a Hilbert space, set $u_i = \sum_k u_{i,k} e_k$ and let $U=(u_{i,j})$.  Then $\{ u_i \}$ is o.n. iff $UU^*= I,$ i.e., $U$ is a coisometry.
\end{prop}

\begin{thm}[The Infinite Dimensional Schmidt Decomposition]\label{thm:inf_dim_schmidt} Let $H$ and $K$ be Hilbert spaces of arbitrary dimension and let $x \in H \otimes K$. Then there are countable orthonormal sets  $ u_k  \in H$ and $v_k \in K$ and $d_k \ge 0$ and $d_k \ge d_{k+1}, \forall k$, such that $x= \sum_k d_k u_k \otimes v_k,$ and so $\|x\|^2= \sum_k d_k^2.$

Moreover, if $x = \sum_k c_k w_k \otimes z_k$ is another such representation of $x$, then $c_k=d_k$ for all $k$. 
\end{thm}
\begin{proof}
  Pick any orthonormal bases $\{ e_{\alpha} : \alpha \in A \}$ and $\{ f_{\beta}: \beta \in B \}$. By Proposition~\ref{prop:tensor_basis}, we can expand $x = \sum_{\alpha, \beta} z_{\alpha, \beta} e_{\alpha} \otimes f_{\beta}$. We know that only countably many of the coefficients are non-zero, so we only need countably many $\alpha$'s and countably many $\beta$'s. So we can write  $x = \sum_{i,j} z_{i,j} e_i \otimes f_j.$ and $\sum_{i,j} |z_{i,j}|^2 = \|x\|^2$.

Let $h_j = \sum_i z_{i, j} e_i$, so that $x = \sum_j h_j\otimes f_j$, and let 
$H = \text{span}\{h_j\}$. Let $\{\ket i\}$ be an orthonormal basis of $H$, and 
write $h_j = \sum_{i} x_{i, j} \ket i$. This gives us 
$x = \sum_{i, j} x_{i, j}\ket i\otimes f_j$.

Let $X = \sum_{i, j} x_{i, j} \ket i\bra j$ be the matrix of a map from $H$ to $H$. Note that $X$ is Hilbert-Schmidt and so compact and also
has dense range because $\text{span}\{\ket i\}  = \text{span} \{h_j\}= \text{span}\{\sum_{i} x_{i, j} \ket i\}$. 
By Proposition~\ref{prop:polar_decomposition}, performing polar decomposition on $X$ yields $X = W |X|$ where 
$W$ is a partial isometry, and $|X| = (X^* X)^{1/2}$. 

Since $|X|$ is compact and positive, it has an orthonormal basis of eigenvectors. This defines a unitary 
$V$ such that $V |X| V^* = D$, where $D  =\sum_{k} d_k \ket k\bra k$ is a diagonal matrix and the $d_k$'s are the singular values of $X$ arranged in decreasing order.  Conjugating $D$ by $V$, we 
get $|X| = V^* D V$. Combining this with the polar decomposition, we get 
$X = W|X| = WV^* D V = U D V$ where $U = WV^*$ is a partial isometry. Moreover, since $X$ has dense range,  $U$ is a coisometry.

Let $U = \sum_{i, j} u_{i,j} \ket i\bra j$, and $V = \sum_{i, j} v_{i, j}\ket i \bra j$. Then, 
\begin{align*}
\sum_{i,j} |x_{i,j}|^2 &=\Tr(X^*X)= \Tr(V^*D^*U^*UDV) \\
&= \Tr(V^*D^2V) = \Tr(D^2) = \sum_k d_k^2.
\end{align*}
Now let $u_k = \sum_{i} u_{i, k} \ket i$ and $v_k = \sum_j v_{k, j} f_j$. Since $U$ is
a coisometry, by Proposition~\ref{prop:coisometry_on}, $\{u_k\}$ is an orthonormal set. Now we have
\[x = \sum_{i, j} x_{i, j} \ket i \otimes f_j = \sum_{i, j, k} u_{i, k} d_k v_{k, j} \ket i \otimes f_j = \sum_{k} d_k u_k\otimes v_k.\]
The statement about the uniqueness of the sequence $d_k$ follows from the fact that these numbers are the singular values of the Hilbert-Schmidt matrix $X$ and that any other choice of basis for representing $X$ would give rise to a matrix that is obtained from $X$ by pre and post multiplying by unitaries, which does not alter the singular values.  
\end{proof}
\section{A primer on C*-algebras}
\label{appendix:cstarprimer}

For readers unfamiliar with C*-algebras, we briefly mention the definitions and tools that we shall use. For very readable general references we recommend \cite{Da} or \cite{KR}.

Given a Hilbert space $\cl H$ we let $B(\cl H)$ denote the set of bounded linear operators from $\cl H$ to $\cl H$.
By a {\it C*-algebra of operators} we mean a subset $\cl A \subseteq B(\cl H)$ for some Hilbert space $\cl H$ satisfying:
\begin{itemize}
\item $X, Y \in \cl A, \lambda \in \bb C \implies (\lambda X+Y) \in \cl A$ and $XY \in \cl A$,
\item $X \in \cl A \implies X^* \in \cl A$, where $X^*$ denotes the adjoint of the operator $X$(sometimes denoted by $X^{\dag}$ in the physics literature),
\item $\cl A$ is closed in the operator norm, i.e., if $X_n \in \cl A$, $X \in B(\cl H)$ and $\|X_n - X \| \to 0,$ then $X \in \cl A$.
\end{itemize}
The first condition is the definition of what it means to say that $\cl A$ is an {\it algebra over the complex field}. The second condition is that $\cl A$ be invariant under the taking of operator
adjoints and the third is that it be a closed subset of $B(\cl H)$ in a certain topology.

C*-algebras of operators also have an abstract characterization. An algebra $\cl A$ over the complex numbers that is equipped with a norm $\| \cdot \|$ that satisfies $\|xy\| \le \|x\| \cdot \|y\|$ is called a {\it normed algebra}. If a normed algebra is {\it complete}, i.e., if every Cauchy sequence converges, then it is called a {\it Banach algebra}.

Given an algebra $\cl A$ over the complex numbers, a {\it *-map} is a map from $\cl A$ to $\cl A$,
$^*: \cl A \to \cl A$ satisfying $(x+y)^*= x^*+y^*$, $(\lambda x)^* = \overline{\lambda} x^*,$ $(xy)^*= y^*x^*$ and $(x^*)^*=x$.

An algebra equipped with a *-map is called a {\it *-algebra}. A map $\pi: \cl A \to \cl B$ between two algebras that is linear and satisfies $\pi(xy) = \pi(x) \pi(y)$ is called a {\it homomorphism}. If both algebras are also *-algebras and the map satisfies $\pi(x^*) = \pi(x)^*$ then $\pi$ is called a {\it *-homomorphism}.

Finally, an { \it (abstract) C*-algebra} $\cl A$ is a Banach *-algebra that satisfies $\|x^*x\|= \|x\|^2, \, \forall x \in \cl A$.
 Note that $B(\cl H)$ is an abstract C*-algebra and so is every C*-algebra of operators.  The celebrated {\it Gelfand-Naimark-Segal theorem} shows that every abstract C*-algebra is in an appropriate sense a C*-algebra of operators.
 
 \begin{thm}[Gelfand-Naimark-Segal] Let $\cl A$ be an abstract C*-algebra. Then there is a Hilbert space $\cl H$ and a map $\pi: \cl A \to B(\cl H)$ such that:
 \begin{itemize}
 \item $\pi$ is a *-homomorphism,
 \item $\|\pi(x)\|= \|x\|$ for all $x \in \cl A$.
 \end{itemize}
 Moreover, if $\cl A$ has a unit element $1 \in \cl A$ then, in addition, one can arrange that $\pi(1) = I_{\cl H}$.
 \end{thm}
  
  A  map satisfying the second condition is called an {\it isometry}. Clearly, an isometry is one-to-one. Conversely, it is a theorem that every one-to-one *-homomorphism is an isometry. A one-to-one, onto *-homomorphism is called a {\it *-isomorphism}.

  The two conditions in the above theorem also guarantee that the range of $\pi,$ $\cl B= \pi(\cl A)$ is a C*-algebra of operators. Thus, $\pi$ is a *-isomorphism from the abstract C*-algebra onto a C*-algebra of operators.
  
  A key element of the proof of the above theorem is their theorem on representations of states.  A {\it state} on an abstract unital C*-algebra $\cl A$ is any linear functional,
  $s: \cl A \to \bb C$ such that $s(1) =1$ and $s(x^*x) \ge 0$ for every $x \in \cl A$.
  
  \begin{thm}[GNS state representation theorem] Let $s:\cl A \to \bb C$ be a state on a unital C*-algebra.  Then there exists a Hilbert space $\cl H_s$, a unit vector $\xi \in \cl H_s$ and a unital *-homomorphism, $\pi_s: \cl A \to B(\cl H_s)$ such that $s(x) = \langle \xi | \pi_s(x) \xi \rangle$ for all $x \in \cl A$ and such that the subspace of vectors of the form $\{ \pi_s(x) \xi : x \in \cl A \}$ is dense in $\cl H_s$.
  \end{thm} 
  
  Thus, the theorem says that for each abstract state there is a way to realize the C*-algebra as a C*-algebra of operators such that the state becomes a vector state. 
  
  \subsection{The State Space}
  Given a unital C*-algebra $\cl A$, the set of all states on $\cl A$, denoted $S(\cl A)$ is a convex set. Moreover, it is endowed with a topology, called the {\it weak*-topology} and in this topology it is a compact set. A net of states $\{ s_{\lambda} \}$ converges to a state $s$ in this topology if and only if $\lim_{\lambda} | s_{\lambda}(a) - s(a)| =0$ for every $a \in \cl A$.


\subsection{Tensor Products of C*-algebras}

Let $\cl A$ and $\cl B$ be two unital C*-algebras, and let $\cl A\otimes \cl B$ be their algebraic tensor product. Given $x= \sum_i a_i \otimes b_i$ and $y = \sum_j c_j \otimes d_j$ in $\cl A \otimes \cl B$ we define their product by
\[ xy = \sum_{i,j} a_i c_j \otimes b_i \otimes d_j,\]
and a *-map by
\[x^* = \sum_i a_i^* \otimes b_i^*.\]
Endowed with these two operations, $\cl A \otimes \cl B$ becomes a *-algebra.

Note that the *-subalgebra $\{ a \otimes 1: a \in \cl A \}$ can be identified with $\cl A$ and similarly, $\{ 1 \otimes b: b \in \cl B \}$ can be identified with $\cl B$. Also $(a \otimes 1)(1 \otimes b) = a \otimes b = (1 \otimes b)(1 \otimes a)$ so that these ``copies" of $\cl A$ and $\cl B$ commute.

There are two important ways to give this *-algebra a norm so that it can be completed to become a C*-algebra. 

Given $x \in \cl A \otimes \cl B$ we set
\[\|x\|_{\max} = \sup \{ \| \pi(x) \|: \ \pi: \cl A \otimes \cl B \to B(\cl H) \text{ is a unital *-homomorphism}\}, \]
where the supremum is taken over all Hilbert spaces $\cl H$ and all unital *-homomorphisms.
The completion of $\cl A \otimes \cl B$ in this norm is a C*-algebra denoted $\cl A \otimes_{\max} \cl B$.

Alternatively, if $\pi_1: \cl A \to B(\cl H_1)$ and $\pi_2: \cl B \to B(\cl H_2)$ are unital *-homomorphisms, then setting $\pi(a \otimes b) = \pi_1(a) \otimes \pi_2(b) \in B(\cl H_1 \otimes \cl H_2)$ and extending linearly, defines a unital *-homomorphism from $\cl A \otimes \cl B$ into $B(\cl H_1 \otimes \cl H_2)$ denoted by $\pi= \pi_1 \otimes \pi_2$.

Given $x \in \cl A \otimes \cl B$ we set
\begin{multline*} \|x\|_{\min} = \sup \{ \| \pi_1 \otimes \pi_2(x) \|: \,\, \pi_1: \cl A \to B(\cl H_1), \, \, \pi_2: \cl B \to B(\cl H_2) \\ \text{ are unital *-homomorphisms}\}.\end{multline*}
The completion of $\cl A \otimes \cl B$ in this norm is a C*-algebra denoted $\cl A \otimes_{\min} \cl B$.

\end{document}